\documentclass[pra,preprint]{revtex4-1}

\usepackage{amsmath,mathrsfs,amsbsy,color,graphicx,bm,amsthm,amsfonts,dsfont}
\usepackage{times}
\usepackage{units}
\usepackage{braket}
\usepackage{amscd}
\usepackage{enumerate}

\newtheorem{theorem}{Theorem}
\newtheorem{lemma}{Lemma}

\begin{document}

\title{Efficient measurement-device-independent detection of multipartite entanglement structure}

\author{Qi Zhao}
\author{Xiao Yuan}
\author{Xiongfeng Ma}
\affiliation{Center for Quantum Information, Institute for Interdisciplinary Information Sciences, Tsinghua University, Beijing, 100084 China}

\begin{abstract}
Witnessing entanglement is crucial in quantum information processing. With properly preparing ancillary states, it has been shown previously that genuine entanglement can be witnessed without trusting measurement devices. In this work, we generalize the scenario and show that generic multipartite entanglement structures, including entanglement of subsystems and entanglement depth, can be witnessed via measurement-device-independent means. As the original measurement-device-independent entanglement witness scheme exploits only one out of four Bell measurement outcomes for each party, a direct generalization to multipartite quantum states will inevitably cause inefficiency in entanglement detection after taking account of statistical fluctuations. To resolve this problem, we also present a way to utilize all the measurement outcomes. The scheme is efficient for multipartite entanglement detection and can be realized with state-of-the-art technologies.

\end{abstract}

\maketitle

\section{Introduction}

Manipulating quantum information provides remarkable advantages in many tasks, including quantum communication and computation \cite{nielsen2010quantum, ladd2010quantum}. It is widely believed that quantum entanglement \cite{Horodecki09} is an essential resource for many quantum information schemes, including Bell nonlocality test \cite{bell1964einstein}, quantum key distribution \cite{bb84, Ekert91}, and quantum computing \cite{nielsen2010quantum}. Hence, witnessing the existence of entanglement is a vital benchmark step for those schemes. A conventional way for witnessing entanglement is by measuring a Hermitian operator $W$ that satisfies $\mathrm{Tr}[\sigma W]\ge 0$ for all separable states $\sigma$ and $\mathrm{Tr}[\rho W]<0$ for a certain entangled state $\rho$. Such a method is generally called entanglement witness (EW) \cite{Terhal200161}. For a review of the subject, see refer to Ref.~\cite{guhne2009} and references therein.

The conclusion of conventional EW relies on faithful realization of measurements. Imperfect measurements can lead to inaccurate estimation of the expected value $\mathrm{Tr}[\rho W]$, which can cause false identification of entanglement even for separable states \cite{Yuan14}. One possible solution to such a problem is by running nonlocality tests \cite{hensen2015loophole, Shalm15, Giustina15}, such as Bell's inequality tests, which can witness entanglement without assuming the realization devices. While realizing a loophole-free Bell test for an arbitrary quantum state is still technically challenging, a compromised method, called measurement-device-independent entanglement witness (MDIEW) is shown to be able to detect arbitrary entangled state \cite{Branciard13} and be experimental friendly \cite{Yuan14,
nawareg2015experimental}. As shown in Fig.~\ref{Fig:MDIEW}, the MDIEW scheme shares a strong similarity to the MDI quantum key distribution protocol \cite{Lo12}, which can also be regarded as a modification of the Bell test \cite{Buscemi12}. In the bipartite scenario, Alice and Bob first prepare ancillary inputs $\tau_s$ and $\omega_t$ according to local random numbers $s$ and $t$, respectively. Then, Alice (resp.~Bob) performs a Bell state measurement (BSM) on the joint state of $\rho_A$ (resp.~$\rho_B$) and the ancillary input $\tau_s$ (resp.~$\omega_t$). Based on the probability distribution of inputs and outputs, it is shown that the witness of entanglement does not rely on the measurement devices.

\begin{figure}[bth]
\centering \resizebox{6cm}{!}{\includegraphics{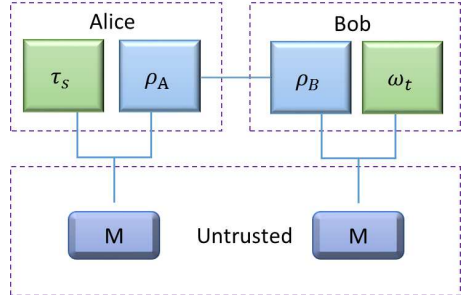}}
\caption{Measurement device independent entanglement witness. Two users, Alice and Bob, are asked to witness the entanglement of a bipartite state $\rho_{AB}$. In the MDIEW scheme, Alice and Bob randomly prepares ancillary states and perform a Bell state measurement jointly on the to-be-witnessed state and the ancillary state.}
\label{Fig:MDIEW}
\end{figure}

For multipartite systems, states can have rich entanglement structures. For instance, when dividing a state into subsystems, how the subsystem entangles with each other determines the \emph{entanglement structure} of the state. Additionally, entanglement structure also have some high-level properties, such as \emph{entanglement depth}, which is related to the concept of $k$-producible states \cite{guhne2005multipartite,seevinck2001sufficient}. A $k$-producible pure state $\ket{\phi}$ can be expressed as a tensor products of subsystems, $\ket{\phi}=\bigotimes_{i=1}^{m}\ket{\phi_i}$, where each subsystem $\ket{\phi_i}$ involves at most $k$ parties. A mixed state is $k$-producible if it can be expressed as a mixture of $k$-producible pure states. If an $N$-partite state is $k$-producible but not $(k-1)$-producible, then such a state has a depth of $k$. When an $N$-partite state has depth of $N$, we call it genuinely entangled. In the original MDIEW scheme \cite{Branciard13}, it is shown that genuine entanglement can be detected in an MDI manner, while, multipartite entanglement apart from genuine entanglement also has important applications in quantum information processing, e.g. high-precision metrology \cite{giovannetti2004quantum} and extreme spin squeezing \cite{sorensen2001entanglement}. Therefore, it is also important to detect general multipartite entanglement structures. Many works have provided ways to detect entanglement relationships between subsystems \cite{huber2013structure,shahandeh2014structural} and entanglement depth \cite{sorensen2001entanglement,liang2015family} with trusted measurement devices. However, it is left open whether one can detect general multipartite entanglement structures, including entanglement between subsystems and entanglement depth, via MDI means.

Also, it is worth mentioning that the original MDIEW protocol is inefficient for detecting multipartite entanglement, especially when the number of parties is large. In the bipartite qubit scenario, only one out of four BSM outcomes of each party is collected for the final estimation of entanglement. As there are four BSM outcomes for each party and in total 16 outcomes for both parties, only a small fraction of experiment data is exploited. When extending the scenario to $N$ parties, only a ratio of $4^{-N}$ outcomes is useful for witness.

In this work, we present an explicit MDI entanglement detection scheme for a multipartite entanglement structure. In Sec.~\ref{Sec:example}, we first review the original MDIEW scheme and point out its inefficiency. In Sec.~\ref{Sec:IMDIEW}, we propose a more efficient MDIEW method that exploits all BSM outcomes to faithfully detect entanglement. As an example, we show how to detect a general two-qubit Werner state. In Sec.~\ref{Sec:MMDIEW}, we show that the efficient MDIEW can be used for detecting multipartite entanglement structure. Finally in Sec.~\ref{Sec:Conclusion}, we discuss our result, its possible application in practice, and prospective works.

\section{MDIEW}\label{Sec:example}
Many efforts have been devoted to detect the existence of entanglement \cite{guhne2009, Dariusz14}. Recently, a Bell-like scenario with quantum inputs was proposed to witness entanglement without trusting the measurement devices, usually referred to as MDIEW \cite{Branciard13}. In the bipartite case, two users, Alice and Bob, share a bipartite state $\rho_{AB}$ defined in a Hilbert space $\mathcal{H_A}\otimes\mathcal{H_B}$ with dimensions $d_A$ and $d_B$. To witness the entanglement of $\rho_{AB}$, Alice and Bob randomly prepare quantum state $\tau_s$ and $\omega_t$, and then perform BSM on the to-be-witnessed state and the ancillary state jointly, respectively. In the original protocol, only one projection outcome is considered as a successful measure, denoted by 1, and other inconclusive outcomes including losses are regarded as a failure, denoted by 0. Conditioned on the input ancillary states, the probability of a successful measurement is denoted by $P(1,1|\tau_s,\omega_t)$,
\begin{equation}\label{Eq:p11}
\begin{aligned}
P(1,1|\tau_s,\omega_t)&=\mathrm{Tr}[(\ket{\phi^+}\bra{\phi^+}\otimes \ket{\phi^+}\bra{\phi^+})(\tau_s\otimes\rho_{AB}\otimes\omega_t)]\\
&=\mathrm{Tr}[(\tau_s^\mathrm{T}\otimes\omega_t^\mathrm{T})\rho_{AB}],
\end{aligned}
\end{equation}
where $\ket{\phi^+}=(\ket{00}+\ket{11})/{\sqrt{2}}$ is a Bell state and corresponds to the selected BSM outcome. The MDIEW value is defined by a linear combination of $P(1,1|\tau_s,\omega_t)$,
\begin{equation}
\begin{aligned}
I(\rho_{AB})=\sum_{s,t}\beta_{s,t}P(1,1|\tau_s,\omega_t).
\end{aligned}
\end{equation}
Here $\beta_{s,t}$ are properly chosen coefficients such that $I(\sigma_{AB})\ge 0$ for any separable state $\sigma_{AB}$ even with arbitrary measurement. Hence, a negative value for $I(\rho_{AB})$ implies nonzero entanglement in $\rho_{AB}$.

In this original scheme, only one measurement outcome is utilized for constructing the MDIEW. When assuming that all BSM outcomes have the same probability, only $1/d_Ad_B$ measurement data are utilized in the bipartite qubit scenario. For a multipartite system with Hilbert space $\mathcal{H}_1 \otimes \mathcal{H}_2 \dots \otimes \mathcal{H}_N $ and $\mathrm{dim} \mathcal{H}_i=d_i$ for $i=1\cdots N$, the fraction of exploited data becomes $(d_1d_2\dots d_N)^{-2}$. Therefore, the original MDIEW scheme will be highly inefficient for detecting multipartite entanglement, especially when statistical fluctuations are taken into consideration.

To be more precise, we consider a practical scenario where an MDIEW experiment for an $N$ partite qubit state runs $G \gg 1$ times. Denote the input ancillary states, the coefficients, and the outcome probability to be $\tau_{1,x_1} \otimes \cdots \otimes\tau_{N,x_N}$, $\beta_{x_1,\cdots x_N}$, and $P(1,\cdots 1|\tau_{1,x_1}\cdots \tau_{N,x_N})$, respectively. Then, the MDIEW value is given by
\begin{equation} \label{Eq:BellvalueN}
\begin{aligned}
I(\rho_{1,2,\dots,N})=\sum_{x_1,\cdots x_N}\beta_{x_1,\cdots x_N}P(1,\cdots 1|\tau_{1,x_1}\cdots \tau_{N,x_N}).
\end{aligned}
\end{equation}
As shown below, the statistical fluctuation with finite size data can be large for a multipartite system.

Denote the observed MDIEW value as $\bar{I}(\rho_{1,2,\dots,N})$. In practice, even if $\bar{I}(\rho_{1,2,\dots,N})$ is negative, due to statistical fluctuations, it is still possible to get it from measuring a separable state when the data size $G$ is finite. Consider the experiment data as a test, then we can use the $p$-value to quantify the probability of getting such a negative value with separable states. Suppose independent and identically distributed data and a large $G$, then the observed probability (rate) $\bar{P}(1,\cdots, 1|\tau_{1,x_1}\cdots \tau_{x_N})$ follows a Gaussian distribution. As the input ancillary states are randomly prepared, the average value $\bar{I}(\rho_{1,2,\dots,N})$ also follows a Gaussian distribution with the expected value defined in Eq.~\eqref{Eq:BellvalueN}. When measuring a separable state, the average value $\bar{I}(\rho_{1,2,\dots,N})$ at least equals $0$ when $G$ goes to infinity. Therefore, we can compute the $p$-value of an observed negative value with $G$ experiment runs. Then, we find that the $p$-value will be in the order of $e^{-G/ \mathrm{O}((d_1d_2\dots d_N)^2)}$. Details of the calculation can be found in the Appendix~\ref{pvaluecal}.

In order to maintain a certain $p$-value, the number of experiment runs $G$ needs to increase exponentially with the number of parties $N$. In the following discussion, we will show that such inefficiency is caused by the poor exploitation of measurement outcomes. By slightly modifying the MDIEW scheme, all measurement outcomes can be utilized.


\section{MDIEW using complete measurement information}\label{Sec:IMDIEW}
In this section, we focus on the bipartite qubit case and show how to construct MDIEW with all measurement outcomes. The method can be naturally extended to the qudit case. MDIEW in the multipartite scenario will be discussed in the next section.

The BSM is defined by projection measurement onto the Bell basis $\{\ket{\phi^{+}},\ket{\phi^{-}},\ket{\psi^{+}},\ket{\psi^{-}}\}$, where
$\ket{\phi^\pm}=\frac{1}{\sqrt{2}}(\ket{00}\pm\ket{11})$ and $\ket{\psi^\pm}=\frac{1}{\sqrt{2}}(\ket{01}\pm\ket{10})$. Label the four BSM outcomes for Alice and Bob by $i$ and $j$ $\in\{1, 2, 3, 4\}$, respectively. Then the probability distribution of outcome $i,j$ given inputs $\tau_s,\omega_t$ can be denoted by $P(i,j|\tau_s,\omega_t)$. 

\begin{theorem} \label{Thm:2MDI}
For every entangled state $\rho_{AB}$, there exists coefficients $\beta_{s,t}^{i,j}$ such that
\begin{equation} \label{Eq:Newwitness}
\begin{aligned}
I(\rho_{AB})=\sum_{s,t,i,j}\beta_{s,t}^{i,j}P(i,j|\tau_s,\omega_t),
\end{aligned}
\end{equation}
where the summation takes over $i,j=1,2,3,4$ and the choices of $s, t$, is an MDIEW for $\rho_{AB}$.
\end{theorem}

\em{Proof.}
In conventional EW, for every entangled state $\rho_{AB}$, there exists a witness $W$  such that $\mathrm{Tr}[W\rho_{AB}]<0$, but $\mathrm{Tr}[W\sigma_{AB}]\ge0$ for any separable state \cite{horodecki1996separability}. The witness $W$ can always be decomposed as a linear combination of a tensor product of local density matrices in $\mathcal{H_A}$ and $\mathcal{H_B}$,
\begin{equation} \label{eq:Wdecomp}
\begin{aligned}
W=\sum_{s,t}\beta_{s,t}(\tau_s)^\mathrm{T}\otimes (\omega_t)^\mathrm{T},
\end{aligned}
\end{equation}
where $\beta_{s,t}$ are real coefficients, $(\tau_s)^\mathrm{T}\in \mathcal{H_A}$, $(\omega_t)^\mathrm{T}\in \mathcal{H_B}$ are density matrices, and $\mathrm{T}$ denotes matrix transpose.
As the transposition map preserves eigenvalues, their transpose $\tau_s$ and $\omega_t$ are also density matrices. Alice and Bob choose ancillary states be the transpose of the bases, $\{\tau_s\}_s$ and $\{\omega_t\}_t$. The conditional probability $P(1,1|\tau_s,\omega_t)$, shown in Eq.~\eqref{Eq:p11}, is proportional to the witness value given by $W$  \cite{Branciard13}.

Now, we need to utilize all the BSM outcomes into the EW. Note that $\ket{\phi^{-}}=\sigma_z\ket{\phi^{+}}$, $\ket{\psi^{+}}=\sigma_x\ket{\phi^{+}}$, $\ket{\psi^{-}}=\sigma_x\sigma_z\ket{\phi^{+}}$, where
$\sigma_z= \left[                 
  \begin{array}{cc}   
    1 & 0\\  
    0 & -1 \\
  \end{array}
\right]    $, $\sigma_x=\left[                 
  \begin{array}{cc}   
    0 & 1\\  
    1 & 0 \\
  \end{array}
\right]   $ are Pauli matrices performed on the second party of Bell states. Define new sets of bases $\tau^i_s$ and $\omega_t^j$, for $i,j\in\{1,2,3,4\}$,
\begin{equation} \label{Define m}
\begin{aligned}
\tau^i_s&=m_i\tau_s m_i^\dagger, \\
\omega_t^j&=m_j\omega_t m_j^\dagger, \\
\end{aligned}
\end{equation}
where $m_1=\mathbb{I},m_2=\sigma_Z,m_3=\sigma_X,m_4=\sigma_X\sigma_Z$. For each $i,j$, the witness $W$ can always be decomposed to
\begin{equation} \label{Eq:WDecom}
\begin{aligned}
W=\sum_{s,t}\beta_{s,t}^{i,j}(\tau^i_s)^\mathrm{T}\otimes (\omega_t^j)^\mathrm{T},
\end{aligned}
\end{equation}
with corresponding real coefficients $\beta_{s,t}^{i,j}$. Note that, for different $i,j$, the coefficients $\beta_{s,t}^{i,j}$ are generally different.

Now, we prove that the witness $I(\rho_{AB})$ defined in Eq.~\eqref{Eq:Newwitness} is an MDIEW with coefficients according to Eq.~\eqref{Eq:WDecom}, in the following two steps. First, we prove the witness to be MDI with the following Lemma.

\begin{lemma}\label{Lemma:MDIly}
The witness value $I(\sigma_{AB})$ is nonnegative for any separable state, $\sigma_{AB}=\sum_x p_x\sigma_A^x\otimes \sigma_B^x$, where $\sum_x p_x=1 $, and an arbitrary measurement $\{A_i\}_{i= 1,2,3,4}\otimes\{B_j\}_{j = 1,2,3,4}$. That is,
\begin{equation}
I(\sigma_{AB})  = \sum_{s,t,i,j}\beta_{s,t}^{i,j}P(i,j|\tau_s,\omega_t) \ge 0.
\end{equation}
\end{lemma}

\begin{proof}
The probability distribution of $P(i,j|\tau_s,\omega_t)$ is given by
\begin{equation} \label{pij}
\begin{aligned}
P(i,j|\tau_s,\omega_t)&=\mathrm{Tr}\left[(A_i\otimes B_j)(\tau_s\otimes\sigma_{AB}\otimes\omega_t)\right]\\
&=\sum_x p_x \mathrm{Tr}[(A_i^x\otimes B_j^x)(\tau_s\otimes \omega_t)]
\end{aligned}
\end{equation}
where $A_i^x=\mathrm{Tr_A}[A_i(\mathbb{I}\otimes \sigma_A^x)]$ and $B_j^x=\mathrm{Tr_B}[B_j(\sigma_B^x\otimes \mathbb{I})]$ and $\mathrm{Tr_A},\mathrm{Tr_B}$ are the partial trace over systems $A, B$, respectively.

Considering the transformation in Eq.~\eqref{Define m}, the probability distribution of $P(i,j|\tau_s,\omega_t)$ can be written as
\begin{equation}
\begin{aligned}
P(i,j|\tau_s,\omega_t)&= \sum_x p_x \mathrm{Tr}[((m_i A_i^x m_i^\dagger) \otimes (m_j  B_j^xm_j^\dagger ))     (\tau_s^i \otimes \omega_t^j)]
\end{aligned}
\end{equation}
Thus, the MDIEW value $I(\sigma_{AB})$ is given by
\begin{equation}
\begin{aligned}
I(\sigma_{AB})&=\sum_{s,t,i,j}\beta_{s,t}^{i,j}P(i,j|\tau_s,\omega_t) \\
&=\sum_{s,t,i,j}\beta_{s,t}^{i,j}\sum_x p_x \mathrm{Tr}\left\{\left[\left(m_i A_i^x m_i^\dagger\right) \otimes\left (m_j  B_j^xm_j^\dagger\right )\right]   \left  (\tau_s^i \otimes
\omega_t^j\right)\right\}\\
&=\sum_{i,j}\sum_x p_x \mathrm{Tr}\left\{\left[\left(m_i A_i^x m_i^\dagger\right) \otimes \left(m_j  B_j^xm_j^\dagger\right )\right]\sum_{s,t} \beta_{s,t}^{i,j}    \left(\tau_s^i \otimes
\omega_t^j\right)\right\}\\
&=\sum_{i,j}\sum_x p_x \mathrm{Tr}\left\{\left[\left(m_i A_i^x m_i^\dagger\right) \otimes \left(m_j  B_j^xm_j^\dagger\right )\right] W^\mathrm{ T }\right \}. \\
\end{aligned}
\end{equation}
Note that for $i,j\in\{1,2,3,4\}$, $(m_i A_i^x m_i^\dagger)^\mathrm{ T }$ and $(m_j  B_j^xm_j^\dagger) ^\mathrm{ T }$ are all positive Hermitian matrix. Then,
$\mathrm{Tr}[((m_i A_i^x m_i^\dagger) \otimes (m_j  B_j^xm_j^\dagger )) W^\mathrm{ T }]=\mathrm{Tr}[W((m_i A_i^x m_i^\dagger)^\mathrm{ T } \otimes (m_j  B_j^xm_j^\dagger )^\mathrm{ T })]\ge 0$.
Consequently, we prove that $I(\sigma_{AB})\ge 0$.
\end{proof}

The second step is to show it to be a witness.
\begin{lemma}\label{Lemma:BSM}
The entanglement of $\rho_{AB}$ can be witnessed when the BSM is faithfully performed.
\end{lemma}
\begin{proof}
When the measurement is perfectly realized, the probability distribution $P(i,j|\tau_s,\omega_t)$ is
\begin{equation}
\begin{aligned}
&P(i,j|\tau_s,\omega_t)\\
&=\mathrm{Tr}\left[\left(m_i\ket{\phi^+}\bra{\phi^+}m_i^\dagger\otimes m_j\ket{\phi^+}\bra{\phi^+}m_j^{\dagger}\right)\times \left(\tau_s\otimes\rho_{AB}\otimes\omega_t\right)\right]\\
&=\frac14\mathrm{Tr}\left[(\tau^i_s)^\mathrm{T}\otimes(\omega_t^j)^\mathrm{T} \rho_{AB}\right].
\end{aligned}
\end{equation}
Then the witness value is
\begin{equation}
\begin{aligned}
I(\rho_{AB})&=\sum_{s,t,i,j}\beta_{s,t}^{i,j}P(i,j|\tau_s,\omega_t)\\
&=\frac14\sum_{i,j}\mathrm{Tr}\left[\sum_{s,t}\beta_{s,t}^{i,j}\left(\tau^i_s\right)^\mathrm{T}\otimes(\omega_t^j)^\mathrm{T} \rho_{AB}\right] \\
&=\frac14\sum_{i,j}\mathrm{Tr}[W\rho_{AB}] \\
&=4\mathrm{Tr}[W\rho_{AB}]<0.
\end{aligned}
\end{equation}
Here, the third equality holds because for each pair of outcome $i,j$, the summation over $s,t$ can construct $W$ according to Eq.~\eqref{Eq:WDecom}. The fourth equality holds because of the summation over all $i,j$ that involves 16 pairs of outcome in total.

\end{proof}

With Lemma~\ref{Lemma:MDIly} and \ref{Lemma:BSM}, we thus show that a negative value of $I(\rho_{AB})$ implies the entanglement of $\rho_{AB}$ even though the measurement devices are not trusted.

\subsection{Example}\label{example}
Now, we will show an example to illustrate the modified MDIEW scheme. We choose a typical state, called the Werner state \cite{PhysRevA.40.4277}, as the target state. The Werner state is defined by a mixture of a maximal entangled state and the  maximal mixed state,
\begin{equation}
\begin{aligned}
\rho_{AB}=p~\ket{\psi^-}\bra{\psi^-}+\frac{1-p}{4}\mathbb{I},
\end{aligned}
\end{equation}
where $\ket{\psi^-}=(\ket{01}-\ket{10})/{\sqrt{2}}$ is singlet state and $p\in[0,1]$. The witness for Werner state is given by,
\begin{equation}
\begin{aligned}
W=\frac{1}{2}-\ket{\psi^-}\bra{\psi^-},
\end{aligned}
\end{equation}
which gives $\mathrm{Tr}[W\rho_{AB}] = (1-3p)/4$. When $p> \frac{1}{3}$, we have $\mathrm{Tr}[W\rho_{AB}] < 0 $, which implies that the Werner state is entangled.

Suppose Alice and Bob choose ancillary states $\tau_1=\omega_1=\frac{\mathbb{I}}{2}$, $\tau_2=\omega_2=\frac{\mathbb{I}+\sigma_x}{2}$, $\tau_3=\omega_3=\frac{\mathbb{I}+\sigma_y}{2}$,
$\tau_4=\omega_4=\frac{\mathbb{I}+\sigma_z}{2}$, where $\mathbb{I}$ is identity and $\sigma_x,\sigma_y,\sigma_z$ are Pauli matrices. The witness $W$
can be decomposed in the basis $(\tau_s^i)^{\mathrm{T}}\otimes (\omega_t^{j})^\mathrm{T} $. For certain outputs $i,j$, the corresponding coefficient matrices $\beta_{s,t}^{i,j}$
are calculated,
\begin{equation}       
\beta^{1,1}_{s,t}=\beta^{2,2}_{s,t}=\beta^{3,3}_{s,t}=\beta^{4,4}_{s,t}=\left[                 
  \begin{array}{cccc}   
    4 & -1 & -1& -1\\  
    -1 & 1 & 0&0\\
    -1 & 0 & 1&0\\
    -1 & 0 & 0&1\\  
  \end{array}
\right],                 
\end{equation}

\begin{equation}       
\beta^{1,2}_{s,t}=\beta^{2,1}_{s,t}=\beta^{3,4}_{s,t}=\beta^{4,3}_{s,t}=\left[                 
  \begin{array}{cccc}   
    0 & -1 & 1& 1\\  
    -1 & 1 & 0&0\\
    1 & 0 & -1&0\\
    1 & 0 & 0&-1\\  
  \end{array}
\right],                 
\end{equation}

\begin{equation}       
\beta^{1,3}_{s,t}=\beta^{2,4}_{s,t}=\beta^{3,1}_{s,t}=\beta^{4,2}_{s,t}=\left[                 
  \begin{array}{cccc}   
    0 & 1 & 1&-1\\  
    1 & -1 & 0&0\\
    1 & 0 & -1&0\\
    -1 & 0 & 0&1\\  
  \end{array}
\right],                
\end{equation}

\begin{equation}       
\beta^{1,4}_{s,t}=\beta^{3,2}_{s,t}=\beta^{2,3}_{s,t}=\beta^{4,1}_{s,t}=\left[                 
  \begin{array}{cccc}   
    0 & 1 & -1&1\\  
    1 & -1 & 0&0\\
    -1 & 0 & 1&0\\
    1 & 0 & 0&-1\\  
  \end{array}
\right].               
\end{equation}
With these coefficients, it is easy to verify that the MDIEW value $I(\rho_{AB})$, given in Eq.~\eqref{Eq:Newwitness}, equals to $1 - 3p$ when the measurement is perfectly realized.

\section{MDI detection of multipartite entanglement structure}\label{Sec:MMDIEW}
In this section, we show that our MDIEW scheme can be applied for efficiently detecting multipartite entanglement structures. First, we focus on applying the modified MDIEW scheme to detect entanglement between subsystems. Then, we extend it to detect a high-level multipartite entanglement property, such as entanglement depth.

\subsection{Detecting entanglement between subsystems}
In the trusted device scenario, the entanglement between subsystems has been well studied \cite{huber2013structure,shahandeh2014structural}. For simplicity, we focus on the case of bipartition $\{A\}\{B\}$ of an $N$-partite state $\rho_{1,2,\cdots,N}$ in the Hilbert space $\mathcal{H}_1\otimes \cdots\otimes \mathcal{H}_N$. Note that the extension from bipartition to multi-partition is a rather natural. $\{A\}\{B\}$ are the two subsystems involving  $k,N-k$ parties, respectively, with $\mathcal{H}_A=\mathcal{H}_1\otimes \cdots\otimes \mathcal{H}_k$ and $\mathcal{H}_B=\mathcal{H}_{k+1}\otimes \cdots\otimes \mathcal{H}_N$. Here we denote $ S_{\{A\}\{B\}}$ to be the set of states that are separable regarding to the partition:
\begin{equation} \label{Define:s}
\begin{aligned}
\mathcal{S}_{\{A\}\{B\}}=\{\rho\in \mathcal{H}_A\otimes \mathcal{H}_B| \rho=\sum_x p_x \rho_A^x\otimes \rho_B^x, \sum p_x=1, \forall x, p_x\ge 0, \rho_A^x\in \mathcal{H}_A, \rho_B^x\in \mathcal{H}_B \}
\end{aligned}
\end{equation}

Furthermore, define a map $\mathcal{M}$,
\begin{equation} \label{mapdefine}
\begin{aligned}
\mathcal{M}(\rho_{1,2,\cdots,N}) =\mathrm{Tr_{\rho_{1,2,\cdots,N}}}[(M_1\otimes\cdots\otimes M_N)(\mathbb{I}_1\otimes \cdots \otimes \mathbb{I}_N\otimes\rho_{1,2,\cdots,N})]
\end{aligned}
\end{equation}
where for each $1\le k \le N$, $M_k$ is a positive-operator valued measure (POVM) acting on the $\mathbb{I}_k$ and the $k$th quantum system of $\rho_{1,2,\cdots,N}$ and $\mathrm{Tr_{\rho_{1,2,\cdots,N}}}$ denotes the partial trace over the space of $\rho_{1,2,\cdots,N}$. Now, we have the following lemma.

\begin{lemma}\label{structuretheorem}
The map $\mathcal{M}$ cannot generate entanglement from separable states between two subsystems ${\{A\}\{B\}}$, that is,
\begin{equation}
\begin{aligned}
\mathcal{M}(\mathcal{S}_{\{A\}\{B\}})\subseteq \mathcal{S}_{\{A\}\{B\}},
\end{aligned}
\end{equation}
where $\mathcal{S}_{\{A\}\{B\}}$ is defined in Eq.~\eqref{Define:s}.

\end{lemma}

\begin{proof}
For any state $\rho \in \mathcal{S}_{\{A\}\{B\}}$, it can be expressed as $\rho=\sum_x p_x \rho_A^x\otimes \rho_B^x $ with $\sum_x p_x=1$, $p_x\ge 0, \forall x$, $\rho_A^x\in \mathcal{H}_A$, and $\rho_B^x\in \mathcal{H}_B$.
In this case, Eq.~\eqref{mapdefine} can be written as
\begin{equation}
\begin{aligned}
\mathcal{M}(\rho)=\sum_x p_x \mathrm{Tr_\rho} \left[(M_A\otimes M_B)(\mathbb{I}_A\otimes\mathbb{I}_B\otimes  \rho_A^x\otimes \rho_B^x)\right],
\end{aligned}
\end{equation}
with $M_A=M_1\otimes \cdots\otimes M_k$, $M_B=M_{k+1}\otimes \cdots\otimes M_N$, $\mathbb{I}_A=\mathbb{I}_1\otimes \cdots\otimes \mathbb{I}_k$ and $\mathbb{I}_B=\mathbb{I}_{k+1}\otimes \cdots\otimes \mathbb{I}_N$.
It can be further written as
\begin{equation}
\begin{aligned}
\mathcal{M}(\rho)=\sum_x p_x M_A^x\otimes M_B^x,
\end{aligned}
\end{equation}
where $M_A^x=\mathrm{Tr_A}\left[M_A\left(\mathbb{I}_A\otimes \rho_A^x\right)\right]$ and $M_B^x=\mathrm{Tr_B}[M_B(\rho_B^x\otimes \mathbb{I}_B)]$, and $\mathrm{Tr_A}$,$\mathrm{Tr_B}$ are
the partial traces over $\mathcal{H_A}$ and $\mathcal{H_B}$ respectively. Note that $M_A$ and $M_B$ are POVMs and $\rho_A^x\in H_A$, $\rho_B^x\in H_B$, so $M_A^x$ and $M_B^x$ are all density matrices, and $\mathcal{M}(\rho)\in \mathcal{S}_{\{A\}\{B\}}$.
\end{proof}

In the trusted device scenario, the entanglement of a bipartition can also be detected with an EW \cite{sperling2013multipartite,brunner2012testing}, $W_{1,2,\dots,N}$. Similar to Eq.~\eqref{eq:Wdecomp}, $W_{1,2,\dots,N}$ can be decomposed to a linear combination of a tense product of local density matrices in $\mathcal{H}_k$ for $1\le k \le N$
\begin{equation}
  W_{1,2,\dots,N} = \sum_{x_1, x_2, \dots, x_N} \beta_{x_1, x_2, \dots, x_N}(\tau_{1, x_1})^\mathrm{T}\otimes (\tau_{2,
  x_2})^\mathrm{T}\otimes\dots\otimes(\tau_{N, x_N})^\mathrm{T}.
\end{equation}
with $x_k \in \{1,2,3,4\}$ where $\beta_{x_1, x_2, \dots, x_N}$ are real coefficients, $(\tau_{k, x_k})^\mathrm{T}\in \mathcal{H}_i$ are density matrices, and $\mathrm{T}$ denotes a matrix transpose. The transpose matrices $\tau_{k, x_k}$ are also density matrices, which are chosen to be the ancillary states for the $k$th party. In order to utilize all the BSM outcomes into the EW, given input ancillary states $\tau_{k, x_k}$ and outcomes $i_k$ for $k$th party, we can decompose $W_{1,2,\dots,N}$ to
\begin{equation}
  W_{1,2,\dots,N} = \sum_{x_1, x_2, \dots, x_N} \beta_{x_1, x_2, \dots, x_N}^{i_1, i_2, \dots, i_N} (\tau_{1, x_1}^{i_1})^\mathrm{T}\otimes (\tau_{2,
  x_2}^{i_2})^\mathrm{T}\otimes\dots\otimes(\tau_{N, x_N}^{i_N})^\mathrm{T}.
\end{equation}
Here, $\tau_{k, x_k}^{i_k}$, with $1\le k\le N$ and $i_k=1,2,3,4$, is defined by
\begin{equation} \label{Define multi m}
\begin{aligned}
\tau_{k, x_k}^{i_k}&=m_{i_k}\tau_{k, x_k} m_{i_k}^\dagger\\
\end{aligned}
\end{equation}
where $m_1=\mathbb{I},m_2=\sigma_z,m_3=\sigma_x,m_4=\sigma_x\sigma_z$. With the coefficients $\beta_{x_1, x_2, \dots, x_N}^{i_1, i_2, \dots, i_N}$, we can define MDIEW for general entanglement structure by
\begin{equation}\label{Eq:MDIEWN}
\begin{aligned}
I(\rho_{1,2,\dots,N}) &= \sum_{i_1,\cdots i_N ,x_1,\cdots x_N }\beta_{x_1, x_2, \dots, x_N}^{i_1, i_2, \dots, i_N}P(i_1,i_2,\cdots i_N|\tau_{1, x_1}\tau_{2,x_2}\cdots \tau_{N,x_N}).
\end{aligned}
\end{equation}

\begin{theorem}\label{structure}
If $W_{1,2,\dots,N}$ detects the entanglement structure of a state $\rho_{1,2,\dots,N}$, $I(\rho_{1,2,\dots,N})$ defined in Eq.~\eqref{Eq:MDIEWN} is an MDIEW for the same structure of
$\rho_{1,2,\dots,N}$.
\end{theorem}

\begin{proof}
Here we focus on the two partition case ${\{A\}\{B\}}$. The proof is similar to the one of Theorem \ref{Thm:2MDI} by extending it to more parties. For a separable state $\rho_{AB}\in \mathcal{S}_{\{A\}\{B\}}$, the probability distribution $P(i_1,i_2,\cdots i_N|\tau_{1, x_1}\tau_{2,x_2}\cdots \tau_{N,x_N})$ is given by
\begin{equation}
\begin{aligned}
&P(i_1,i_2,\cdots i_N|\tau_{1, x_1}\tau_{2,x_2}\cdots \tau_{N,x_N})\\
&=  \mathrm{Tr}\left[\bigotimes_{k=1}^N M_k \left(\bigotimes_{k=1}^N \tau_{k, x_k} \otimes \rho_{AB}\right)\right]\\
&=\mathrm{Tr}\left\{\mathrm{Tr}_{\rho}\left[\bigotimes_{k=1}^N M_i \left(\bigotimes_{k=1}^N \mathbb{I}_k \otimes\rho_{AB}\right)\right] \bigotimes_{k=1}^N \tau_{k, x_k}\right\}\\
&=\mathrm{Tr}\left[\bigotimes_{k=1}^N \tau_{k, x_k}\mathcal{M}\left(\rho_{AB}\right)\right]
\end{aligned}
\end{equation}
where  $\mathcal{M}(\rho_{AB}) $ is the map defined in Eq.~\eqref{mapdefine}.
Thus
\begin{equation}
\begin{aligned}
 I(\rho_{AB})=&\sum_{i_1,\cdots i_N ,x_1,\cdots x_N }\beta_{x_1, x_2, \dots, x_N}^{i_1, i_2, \dots, i_N}  \bigotimes_{k=1}^N \tau_{k, x_k}\mathcal{M}(\rho_{AB})\\
 &=\sum_{i_1,\cdots i_N} \mathrm{Tr}[W \mathcal{M}(\rho_{AB})]
\end{aligned}
\end{equation}
According to Lemma~\ref{structuretheorem}, we have that $\mathcal{M}(\rho_{AB})\in \mathcal{S}_{\{A\}\{B\}}$. Thus, we prove that $I(\rho_{AB})\ge 0$ for all $\rho_{AB}\in \mathcal{S}_{\{A\}\{B\}}$.

To show $I(\rho_{1,2,\dots,N})$ to be a witness for $\rho_{1,2,\dots,N}$ with ideal measurements, we refer to Lemma~\ref{idealmeasure}.
\end{proof}

\begin{lemma}\label{idealmeasure}
The entanglement structure of $\rho_{1,2,\dots,N}$ can be witnessed by $I(\rho_{1,2,\dots,N})$ when the BSM is faithfully realized.
\end{lemma}
\begin{proof}
When the BSM is faithfully performed, the probability distribution  $P(i_1,i_2,\cdots i_N|\tau_{1, x_1}\tau_{2,x_2}\cdots \tau_{N,x_N})$ is given by,
\begin{equation}
\begin{aligned}
&P(i_1,i_2,\cdots i_N|\tau_{1, x_1}\tau_{2,x_2}\cdots \tau_{N,x_N})\\
&=\mathrm{Tr}\left[\left(\bigotimes_{k=1}^N m_{i_k}\ket{\phi^+}\bra{\phi^+}m_{i_k}^\dagger\right) \times \left(\bigotimes_{k=1}^N \tau_{k, x_k} \otimes\rho_{1,2,\dots,N}\right)\right]\\
&=2^{-N}\mathrm{Tr}\left[\bigotimes_{k=1}^N \left(\tau_{k, x_k}^{i_k}\right)^\mathrm{T} \rho_{1,2,\dots,N}\right].\\
\end{aligned}
\end{equation}
Then,
\begin{equation}
\begin{aligned}
&I(\rho_{1,2,\dots,N})\\
&=\sum_{i_1,\cdots i_N ,x_1,\cdots x_N }\beta_{x_1, x_2, \dots, x_N}^{i_1, i_2, \dots, i_N} \mathrm{Tr}\left[\bigotimes_{k=1}^N \left(\tau_{k, x_k}^{i_k}\right)^\mathrm{T}
\rho_{1,2,\dots,N}\right]/2^N\\
&=\sum_{i_1,\cdots i_N}\mathrm{Tr}\left[W_{1,2,\dots,N}\rho_{1,2,\dots,N}\right]/2^N\\
&=2^N \mathrm{Tr}\left[W_{1,2,\dots,N}\rho_{1,2,\dots,N}\right]<0.
\end{aligned}
\end{equation}

\end{proof}

Similar to Eq.~\eqref{Define:s}, we can also define the multi-partition states. The proofs of Lemma \ref{structuretheorem} and Theorem \ref{structure} mainly focus on the bipartition case, but can be
extended to multi-partition cases naturally. Notice that Lemma \ref{idealmeasure} is in general independent of total party number $N$ and entanglement structure. As long as there exists a witness $W$
, then $I(\rho_{1,2,\dots,N})$ is a witness under the ideal measurements assumption.

\subsection{Detecting entanglement depth}
Besides the entanglement of subsystems, there are other high-level entanglement properties for multipartite quantum states, such as \emph{entanglement depth} \cite{sorensen2001entanglement,liang2015family}. There exists a conventional witness $W_{1,2,\dots,N}$ for detecting the depth of a quantum state \cite{sorensen2001entanglement,liang2015family}. Following a similar way of detecting entanglement structure of subsystems, one can define an MDIEW for detecting entanglement depth similar to Eq.~\eqref{Eq:MDIEWN}. Then, according to Lemma \ref{idealmeasure}, one can easily see that it is indeed a witness for depth when the measurement is ideally realized. Now, we need to prove that such a witness is MDI with the following lemma.

\begin{lemma}\label{depththeorem}
The map $\mathcal{M}$, defined in Eq.~\eqref{mapdefine}, cannot increase the entanglement depth.
\end{lemma}
\begin{proof}
An $N$-partite state $\rho$ that has $k$-depth entanglement  can be expressed as follows:
\begin{equation}
\begin{aligned}
\rho=\sum_x p_x \bigotimes_{i=1}^{m_x} \rho_i^x,
\end{aligned}
\end{equation}
where $m_x\le N$, $\sum_x p_x=1$, and for any $x$, $p_x\ge 0$, $\sum_x p_x=1$, and for every $i$ and $x$, the state $\rho_i^x$ contains at most $k$ parties. After the map $\mathcal{M}$, we have
\begin{equation}
\begin{aligned}
\mathcal{M}(\rho)&=\sum_x p_x \mathrm{Tr_\rho}[(M_1\otimes\cdots\otimes M_n)(\mathbb{I}_1\otimes \cdots \otimes \mathbb{I}_N\otimes\rho)]\\
&=\sum_x p_x\mathrm{Tr_\rho}\left[\left(\bigotimes_{i=1}^{m_x}M_i^x\right)\left(\bigotimes_{i=1}^{m_x}\mathbb{I}_i^x \otimes \bigotimes_{i=1}^{m_x} \rho_i^x\right)\right]\\
&=\sum_x p_x \bigotimes_{i=1}^{m_x}\sigma _i^x,
\end{aligned}
\end{equation}
where $\sigma_i^x=\mathrm{Tr_{\rho_i^x}}\left[M_i^x\left(\mathbb{I}_i^x\otimes \rho_i^x\right)\right]$ is a positive Hermitian matrices and involves at most $k$ parties. Thus
$\mathcal{M}(\rho)$ is at most $k$-depth entangled.
\end{proof}

\begin{theorem}
If $W_{1,2,\dots,N}$ detects entanglement depth for state $\rho_{1,2,\dots,N}$, then $I(\rho_{1,2,\dots,N})$ defined in Eq.~\eqref{Eq:MDIEWN} is an MDIEW for $\rho_{1,2,\dots,N}$.
\end{theorem}
\begin{proof}
We skip the proof, since it is very similar to the one for Theorem \ref{structure}, where we only need to replace Lemma \ref{structuretheorem} with Lemma \ref{depththeorem}.
\end{proof}

In summary, our MDI scheme can detect entanglement structure and entanglement depth. In particular, when an $N$-partite quantum state has a depth of $N$, it is also called genuinely entangled. Thus, our scheme can also be used for detecting genuine entanglement.


\subsection{Example}
Here, we show an explicit example to illustrate the MDI entanglement depth detection method. We consider a mixture of the tripartite $W$-state and white noise as the target state,
\begin{equation}\label{eq:target}
\begin{aligned}
\rho=p~\ket{\psi_W}\bra{\psi_W}+\frac{1-p}{8}\mathbb{I},
\end{aligned}
\end{equation}
where $\ket{\psi_W}=(\ket{001}+\ket{010}+\ket{100})/{\sqrt{3}}$ is the $W$-state. To detect the entanglement of this state, we utilize a witness
\begin{equation}
\begin{aligned}
W=\alpha-\ket{\psi_W}\bra{\psi_W},
\end{aligned}
\end{equation}
which gives an average value of
\begin{equation}\label{}
\mathrm{Tr}[W\rho] =\alpha - \frac{1}{8}-\frac{7}{8}p.
\end{equation}

For different values of $\alpha$, it is shown in Ref.~\cite{sperling2013multipartite} that the witness can detect different entanglement depths of the state. For example, when $\alpha = 2/3$, a negative value $\mathrm{Tr}[W\rho]<0$ would indicate $\rho$ to be genuinely entangled. That is, its entanglement depth is three. When $\alpha = 4/9$, a negative value $\mathrm{Tr}[W\rho]<0$ will indicate $\rho$ to be entangled instead of fully separable. That is, its entanglement depth is at least two. The target state, defined in Eq.~\eqref{eq:target}, is genuinely entangled with a depth of three when $p > 13/21$; and it is not fully separable with a depth at least two when $p>23/63$.

Now, we show the MDI detection for entanglement depth. Suppose that the ancillary input states are $\tau_{k,1}=\frac{\mathbb{I}}{2}$, $\tau_{k,2}=\frac{\mathbb{I}+\sigma_x}{2}$, $\tau_{k,3}=\frac{\mathbb{I}+\sigma_y}{2}$,
$\tau_{k,4}=\frac{\mathbb{I}+\sigma_z}{2}$, where the indexes $k=\{1,2,3\}$ denote the three parties, $\mathbb{I}$ is the identity matrix, and $\sigma_x,\sigma_y,\sigma_z$ are Pauli matrices. The witness $W$ can be decomposed in the basis $(\tau_{1,x_1}^{i_1})^{\mathrm{T}}\otimes (\tau_{2,x_2}^{i_2})^{\mathrm{T}}\otimes (\tau_{3,x_3}^{i_3})^{\mathrm{T}}$.  Here for simplicity, we show only the coefficient matrices $\beta_{x_1,x_2,x_3}^{i_1,i_2,i_3}$ for certain outputs $i_1=i_2=i_3=1$,
\begin{equation}       
\beta^{1,1,1}_{x_1,x_2,1}=\left[                 
  \begin{array}{cccc}   
    4/9 &  0 &  0&  0\\  
     0 &  0 & 0&2/3\\
     0 & 0 & 0&-2/3\\
     0 & 2/3 & -2/3&-2/3\\  
  \end{array}
\right],                 
\end{equation}
\begin{equation}       
\beta^{1,1,1}_{x_1,x_2,2}=\left[                 
  \begin{array}{cccc}   
    0 & 0 & 0&2/3\\  
    0 & 0 & 0&-2/3\\
    0 & 0 & 0&0\\
    2/3 & -2/3& 0&0\\  
  \end{array}
\right],                 
\end{equation}
\begin{equation}       
\beta^{1,1,1}_{x_1,x_2,3}=\left[                 
  \begin{array}{cccc}   
    0 & 0 & 0&-2/3\\  
    0 & 0 & 0&0\\
    0 & 0 & 0&2/3\\
   -2/3 & 0 & 2/3&0\\  
  \end{array}
\right],                 
\end{equation}
\begin{equation}       
\beta^{1,1,1}_{x_1,x_2,4}=\left[                 
  \begin{array}{cccc}   
    -4/3 & 2/3 & 2/3& -2/3\\  
    2/3 & -2/3 & 0&0\\
    2/3 & 0 & -2/3&0\\
    -2/3 & 0 & 0&1\\  
  \end{array}
\right],                 
\end{equation}
where the matrix indexes run over different values of $x_1$ and $x_2$. For the other outputs, the witness $W$ can be similarly decomposed in the other base according to Eq.~\eqref{Define multi m} and the coefficient matrices $\beta_{x_1,x_2,x_3}^{i_1,i_2,i_3}$ for other outputs can be obtained in a similar way. In total, the scheme involves $64$ different outputs cases. With these coefficients, we can verify that the MDIEW value $I(\rho)$, given in Eq.~\eqref{Eq:MDIEWN}, equals $8\alpha - 1-7p$ when measurement is perfectly realized. For different values of $\alpha$, a negative MDIEW value of $I(\rho)$ can be used to detect the entanglement depth.

\section{Discussion} \label{Sec:Conclusion}

In this work, we propose an efficient measurement-device-independent entanglement witness scheme that can be applied for detecting multipartite entanglement structures. Compared to the original proposal \cite{Branciard13}, which cannot detect multipartite entanglement efficiently, we make use of all measurement outcomes for overcoming this problem. Furthermore, we show that our scheme can detect complex entanglement structures, including entanglement between subsystems and entanglement depth. Our result can be applied to the state-of-art experiment for witnessing multipartite entanglement without trusting the measurement devices.

Recently, improved MDIEW schemes that maximally exploit the experiment data have been proposed \cite{PhysRevA.93.042317, PhysRevLett.116.190501}. In these schemes, one can additionally run a post-processing program to find the optimal coefficient that minimize the MDIEW value given the probability distribution. In this case, all measurement outcomes can be maximally exploited after the optimization. However, although the optimization works efficiently for small-scale systems, it will become exponentially hard with increasing number of parties. Thus, how to apply the optimal scheme for efficiently detecting multipartite entanglement is an interesting prospective project. Conventional EW is originally designed to efficiently detect the entanglement of states. As our MDIEW is based on conventional EW, it can be used for efficient and practical entanglement detection. Whether the combination of these two methods will lead to a better performance is also an
interesting open problem.


In our MDIEW scheme, the ancillary states for each party should form a basis for Hermitian operators, in which all witness operator can be decomposed. The Hermitian operator for qubits has a basis that consists of four elements. In this case, when the input ancillary states are independently prepared for each party, there are at least $4^N$ different types of inputs, which is exponential to the number of parties $N$. Such a problem can be resolved by noticing a beautiful property of MDIEW found in Ref.~\cite{rosset2013entangled}: the MDIEW scheme is valid even though shared randomness and classical communications are allowed. In this case, as long as the input ancillary states as a whole are randomly prepared, the MDIEW scheme will be reliable. It has been shown that we only need to randomly prepare input ancillary states for each party without worrying about whether the sample size is large enough for all different input conditions \cite{Yuan14}. However, it is still an interesting open question to see whether the number of different input ancillary states that defines an MDIEW can be polynomial to the number of parties $N$.

In the Bell test, three famous loopholes should be closed for guaranteeing a faithful violation of Bell inequality \cite{Brunner14}. The \emph{locality loophole} requires that different parties are sufficiently separated such that they cannot signaling. The \emph{efficiency loophole} requires that the detection efficiency should be larger than a certain threshold \cite{Massar03, Wilms08}. The \emph{randomness or freewill loophole} requires that the input randomness need to be random enough \cite{Koh12,Yuan15,Yuan152}. When the three loopholes are closed, a faithful violation of Bell inequality can witness the existence of entanglement. In the MDIEW scheme, we can see that the locality and efficiency loopholes are not more problems any longer \cite{rosset2013entangled}. On the other hand, it is still meaningful to discuss the randomness or freewill loophole. In one extreme case where all the inputs are perfectly random, the MDIEW is secure; while in the other extreme case where the inputs are all pre-determined, the MDIEW becomes unreliable. Therefore, it would be an interesting question to investigate the randomness requirement that guarantees the security of the MDIEW scheme.


We show that MDIEW can efficiently detect multipartite entanglement stricture. It would be  interesting to see whether more complex entanglement properties can be detected in an MDI manner. For instance, it is well known that multipartite entanglement can be categorized into different classes under stochastic local operations and classical communication
\cite{Dur00,borsten2010four}. Conventional witness can be used for detecting a different entanglement class \cite{acin2001classification}. The key of the MDIEW for entanglement structure is that the map $\mathcal{M}$ defined in Eq.~\eqref{mapdefine} can not generate entanglement. As the MDIEW scheme allows classical communication, in transforming a conventional EW to an MDI one may not change the entanglement class intuitively. However, it is still an open question to design MDIEW for entanglement classification.

\section*{Acknowledgments}
The authors acknowledge insightful discussions with Y.~Liang. This work was supported by the 1000 Youth Fellowship program in China. Z.~Q.~and X.~Y.~contributed equally to this work.

\appendix
\section{Calculation of the $p$-value }\label{pvaluecal}

Given an observed negative average MDIEW value $\bar{I}(\rho_{1,2,\dots,N})$, we need to avoid reaching the wrong conclusion. In statistics, we can apply a $p$-value to quantify the probability of getting such a negative value with separable states. Under the assumption of independent and identically distributed data and large $G$,  $\bar{P}(1,\cdots, 1|\tau_{1,x_1}\cdots \tau_{x_N})$, $\bar{I}(\rho_{1,2,\dots,N})$ both follow the Gaussian distribution. Therefor, the $p$-value of an observed negative value with $G$ experiment runs is:
\begin{equation}\label{eq:pvalue}
p = e^{-M\left(\bar{I}(\rho_{1,2,\dots,N}))\right)^2/(2\sigma^2)}.
\end{equation}
Here, $\sigma$ is the standard deviation of ${I}(\rho_{1,2,\dots,N})$, which is given by
\begin{equation}
\sigma = \sqrt{\sum_{x_1,\cdots x_N} \beta_{x_1,\cdots x_N}^2\sigma^2(\tau_{1,x_1}\cdots \tau_{N,x_N})},
\end{equation}
where $\sigma(\tau_{1,x_1}\cdots \tau_{x_N})$ is the standard deviation of  $P(1,\cdots 1|\tau_{1,x_1}\cdots \tau_{N,x_N})$ and can be expressed as follows:
\begin{equation}
\sigma(\tau_{1,x_1}\cdots \tau_{x_N})=\sqrt{P(1,\cdots 1|\tau_{1,x_1}\cdots
\tau_{N,x_N})\left(1-P(1,\cdots 1|\tau_{1,x_1}\cdots \tau_{N,x_N})\right)}
\end{equation}

The probability distribution is
\begin{equation}
\begin{aligned}
&P(1,\cdots 1|\tau_{1,x_1}\cdots \tau_{x_N})\\
&=\mathrm{Tr}[(\ket{\phi^+}\bra{\phi^+}\otimes \cdots \otimes \ket{\phi^+}\bra{\phi^+})(\tau_{1,x_1} \otimes \cdots  \tau_{x_N})\otimes \rho_{1,2,\dots,N})]\\
&=\mathrm{Tr}[(\tau_{1,x_1}^{\mathrm{T}} \otimes \cdots  \otimes\tau_{x_N}^{\mathrm{T}}) \rho_{1,2,\dots,N})]/(d_1d_2\dots d_N)
\end{aligned}
\end{equation}
where $\ket{\phi^+} = \frac{1}{\sqrt{2}}(\ket{00} + \ket{11})$, and $d_i$ and $\tau_{i,x_i} (i=1,\cdots,N)$ are the dimension and ancillary state for the $i$th party, respectively.

For a randomly chosen state, $\rho_{1,2,\dots,N} = I_{d_1d_2\dots d_N}/(d_1d_2\dots d_N)$, where $I_{d_1d_2\dots d_N}$ is the identity matrix with size $d_1d_2\dots d_N$, we further have that
\begin{equation}
\begin{aligned}
  P(1,\cdots 1|\tau_{1,x_1}\cdots \tau_{x_N}) &= \mathrm{Tr}[(\tau_{1,x_1}^{\mathrm{T}} \otimes \cdots  \otimes\tau_{x_N}^{\mathrm{T}}) I_{d_1d_2\dots d_N}/(d_1d_2\dots d_N)]/(d_1d_2\dots d_N) \\
  &= 1/(d_1d_2\dots d_N)^2
\end{aligned}
\end{equation}

Thus roughly speaking, we have that
\begin{equation}
P(1,\cdots 1|\tau_{1,x_1}\cdots \tau_{N,x_N}) \sim 1/(d_1d_2\dots d_N)^2.
\end{equation}
Consequently, we find that the $p$-value is an order of $e^{-G/ \mathrm{O}((d_1d_2\dots d_N)^2)}$.

\bibliography{bibMDIEW}

\begin{thebibliography}{40}%
\makeatletter
\providecommand \@ifxundefined [1]{%
 \@ifx{#1\undefined}
}%
\providecommand \@ifnum [1]{%
 \ifnum #1\expandafter \@firstoftwo
 \else \expandafter \@secondoftwo
 \fi
}%
\providecommand \@ifx [1]{%
 \ifx #1\expandafter \@firstoftwo
 \else \expandafter \@secondoftwo
 \fi
}%
\providecommand \natexlab [1]{#1}%
\providecommand \enquote  [1]{``#1''}%
\providecommand \bibnamefont  [1]{#1}%
\providecommand \bibfnamefont [1]{#1}%
\providecommand \citenamefont [1]{#1}%
\providecommand \href@noop [0]{\@secondoftwo}%
\providecommand \href [0]{\begingroup \@sanitize@url \@href}%
\providecommand \@href[1]{\@@startlink{#1}\@@href}%
\providecommand \@@href[1]{\endgroup#1\@@endlink}%
\providecommand \@sanitize@url [0]{\catcode `\\12\catcode `\$12\catcode
  `\&12\catcode `\#12\catcode `\^12\catcode `\_12\catcode `\%12\relax}%
\providecommand \@@startlink[1]{}%
\providecommand \@@endlink[0]{}%
\providecommand \url  [0]{\begingroup\@sanitize@url \@url }%
\providecommand \@url [1]{\endgroup\@href {#1}{\urlprefix }}%
\providecommand \urlprefix  [0]{URL }%
\providecommand \Eprint [0]{\href }%
\providecommand \doibase [0]{http://dx.doi.org/}%
\providecommand \selectlanguage [0]{\@gobble}%
\providecommand \bibinfo  [0]{\@secondoftwo}%
\providecommand \bibfield  [0]{\@secondoftwo}%
\providecommand \translation [1]{[#1]}%
\providecommand \BibitemOpen [0]{}%
\providecommand \bibitemStop [0]{}%
\providecommand \bibitemNoStop [0]{.\EOS\space}%
\providecommand \EOS [0]{\spacefactor3000\relax}%
\providecommand \BibitemShut  [1]{\csname bibitem#1\endcsname}%
\let\auto@bib@innerbib\@empty
\bibitem [{\citenamefont {Nielsen}\ and\ \citenamefont
  {Chuang}(2010)}]{nielsen2010quantum}%
  \BibitemOpen
  \bibfield  {author} {\bibinfo {author} {\bibfnamefont {M.~A.}\ \bibnamefont
  {Nielsen}}\ and\ \bibinfo {author} {\bibfnamefont {I.~L.}\ \bibnamefont
  {Chuang}},\ }\href@noop {} {\emph {\bibinfo {title} {Quantum computation and
  quantum information}}}\ (\bibinfo  {publisher} {Cambridge university press},\
  \bibinfo {year} {2010})\BibitemShut {NoStop}%
\bibitem [{\citenamefont {Ladd}\ \emph {et~al.}(2010)\citenamefont {Ladd},
  \citenamefont {Jelezko}, \citenamefont {Laflamme}, \citenamefont {Nakamura},
  \citenamefont {Monroe},\ and\ \citenamefont {O¡¯Brien}}]{ladd2010quantum}%
  \BibitemOpen
  \bibfield  {author} {\bibinfo {author} {\bibfnamefont {T.~D.}\ \bibnamefont
  {Ladd}}, \bibinfo {author} {\bibfnamefont {F.}~\bibnamefont {Jelezko}},
  \bibinfo {author} {\bibfnamefont {R.}~\bibnamefont {Laflamme}}, \bibinfo
  {author} {\bibfnamefont {Y.}~\bibnamefont {Nakamura}}, \bibinfo {author}
  {\bibfnamefont {C.}~\bibnamefont {Monroe}}, \ and\ \bibinfo {author}
  {\bibfnamefont {J.~L.}\ \bibnamefont {O¡¯Brien}},\ }\href@noop {} {\bibfield
  {journal} {\bibinfo  {journal} {Nature}\ }\textbf {\bibinfo {volume} {464}},\
  \bibinfo {pages} {45} (\bibinfo {year} {2010})}\BibitemShut {NoStop}%
\bibitem [{\citenamefont {Horodecki}\ \emph {et~al.}(2009)\citenamefont
  {Horodecki}, \citenamefont {Horodecki}, \citenamefont {Horodecki},\ and\
  \citenamefont {Horodecki}}]{Horodecki09}%
  \BibitemOpen
  \bibfield  {author} {\bibinfo {author} {\bibfnamefont {R.}~\bibnamefont
  {Horodecki}}, \bibinfo {author} {\bibfnamefont {P.}~\bibnamefont
  {Horodecki}}, \bibinfo {author} {\bibfnamefont {M.}~\bibnamefont
  {Horodecki}}, \ and\ \bibinfo {author} {\bibfnamefont {K.}~\bibnamefont
  {Horodecki}},\ }\href {\doibase 10.1103/RevModPhys.81.865} {\bibfield
  {journal} {\bibinfo  {journal} {Rev. Mod. Phys.}\ }\textbf {\bibinfo {volume}
  {81}},\ \bibinfo {pages} {865} (\bibinfo {year} {2009})}\BibitemShut
  {NoStop}%
\bibitem [{\citenamefont {Bell}(1964)}]{bell1964einstein}%
  \BibitemOpen
  \bibfield  {author} {\bibinfo {author} {\bibfnamefont {J.}~\bibnamefont
  {Bell}},\ }\href@noop {} {\bibfield  {journal} {\bibinfo  {journal}
  {Physics}\ }\textbf {\bibinfo {volume} {1}},\ \bibinfo {pages} {195}
  (\bibinfo {year} {1964})}\BibitemShut {NoStop}%
\bibitem [{\citenamefont {Bennett}\ and\ \citenamefont
  {Brassard}(1984)}]{bb84}%
  \BibitemOpen
  \bibfield  {author} {\bibinfo {author} {\bibfnamefont {C.~H.}\ \bibnamefont
  {Bennett}}\ and\ \bibinfo {author} {\bibfnamefont {G.}~\bibnamefont
  {Brassard}},\ }in\ \href@noop {} {\emph {\bibinfo {booktitle} {Proceedings of
  the IEEE International Conference on Computers, Systems and Signal
  Processing}}}\ (\bibinfo  {publisher} {IEEE Press},\ \bibinfo {address} {New
  York},\ \bibinfo {year} {1984})\ pp.\ \bibinfo {pages} {175--179}\BibitemShut
  {NoStop}%
\bibitem [{\citenamefont {Ekert}(1991)}]{Ekert91}%
  \BibitemOpen
  \bibfield  {author} {\bibinfo {author} {\bibfnamefont {A.~K.}\ \bibnamefont
  {Ekert}},\ }\href {\doibase 10.1103/PhysRevLett.67.661} {\bibfield  {journal}
  {\bibinfo  {journal} {Phys. Rev. Lett.}\ }\textbf {\bibinfo {volume} {67}},\
  \bibinfo {pages} {661} (\bibinfo {year} {1991})}\BibitemShut {NoStop}%
\bibitem [{\citenamefont {Terhal}(2001)}]{Terhal200161}%
  \BibitemOpen
  \bibfield  {author} {\bibinfo {author} {\bibfnamefont {B.~M.}\ \bibnamefont
  {Terhal}},\ }\href {\doibase http://dx.doi.org/10.1016/S0024-3795(00)00251-2}
  {\bibfield  {journal} {\bibinfo  {journal} {Linear Algebra and its
  Applications}\ }\textbf {\bibinfo {volume} {323}},\ \bibinfo {pages} {61 }
  (\bibinfo {year} {2001})}\BibitemShut {NoStop}%
\bibitem [{\citenamefont {G{\"u}hne}\ and\ \citenamefont
  {T{\'o}th}(2009)}]{guhne2009}%
  \BibitemOpen
  \bibfield  {author} {\bibinfo {author} {\bibfnamefont {O.}~\bibnamefont
  {G{\"u}hne}}\ and\ \bibinfo {author} {\bibfnamefont {G.}~\bibnamefont
  {T{\'o}th}},\ }\href@noop {} {\bibfield  {journal} {\bibinfo  {journal}
  {Physics Reports}\ }\textbf {\bibinfo {volume} {474}},\ \bibinfo {pages} {1}
  (\bibinfo {year} {2009})}\BibitemShut {NoStop}%
\bibitem [{\citenamefont {Xu}\ \emph {et~al.}(2014)\citenamefont {Xu},
  \citenamefont {Yuan}, \citenamefont {Chen}, \citenamefont {Lu}, \citenamefont
  {Yao}, \citenamefont {Ma}, \citenamefont {Chen},\ and\ \citenamefont
  {Pan}}]{Yuan14}%
  \BibitemOpen
  \bibfield  {author} {\bibinfo {author} {\bibfnamefont {P.}~\bibnamefont
  {Xu}}, \bibinfo {author} {\bibfnamefont {X.}~\bibnamefont {Yuan}}, \bibinfo
  {author} {\bibfnamefont {L.-K.}\ \bibnamefont {Chen}}, \bibinfo {author}
  {\bibfnamefont {H.}~\bibnamefont {Lu}}, \bibinfo {author} {\bibfnamefont
  {X.-C.}\ \bibnamefont {Yao}}, \bibinfo {author} {\bibfnamefont
  {X.}~\bibnamefont {Ma}}, \bibinfo {author} {\bibfnamefont {Y.-A.}\
  \bibnamefont {Chen}}, \ and\ \bibinfo {author} {\bibfnamefont {J.-W.}\
  \bibnamefont {Pan}},\ }\href {\doibase 10.1103/PhysRevLett.112.140506}
  {\bibfield  {journal} {\bibinfo  {journal} {Phys. Rev. Lett.}\ }\textbf
  {\bibinfo {volume} {112}},\ \bibinfo {pages} {140506} (\bibinfo {year}
  {2014})}\BibitemShut {NoStop}%
\bibitem [{\citenamefont {Hensen}\ \emph {et~al.}(2015)\citenamefont {Hensen},
  \citenamefont {Bernien}, \citenamefont {Dr{\'e}au}, \citenamefont {Reiserer},
  \citenamefont {Kalb}, \citenamefont {Blok}, \citenamefont {Ruitenberg},
  \citenamefont {Vermeulen}, \citenamefont {Schouten}, \citenamefont
  {Abell{\'a}n} \emph {et~al.}}]{hensen2015loophole}%
  \BibitemOpen
  \bibfield  {author} {\bibinfo {author} {\bibfnamefont {B.}~\bibnamefont
  {Hensen}}, \bibinfo {author} {\bibfnamefont {H.}~\bibnamefont {Bernien}},
  \bibinfo {author} {\bibfnamefont {A.}~\bibnamefont {Dr{\'e}au}}, \bibinfo
  {author} {\bibfnamefont {A.}~\bibnamefont {Reiserer}}, \bibinfo {author}
  {\bibfnamefont {N.}~\bibnamefont {Kalb}}, \bibinfo {author} {\bibfnamefont
  {M.}~\bibnamefont {Blok}}, \bibinfo {author} {\bibfnamefont {J.}~\bibnamefont
  {Ruitenberg}}, \bibinfo {author} {\bibfnamefont {R.}~\bibnamefont
  {Vermeulen}}, \bibinfo {author} {\bibfnamefont {R.}~\bibnamefont {Schouten}},
  \bibinfo {author} {\bibfnamefont {C.}~\bibnamefont {Abell{\'a}n}},  \emph
  {et~al.},\ }\href@noop {} {\bibfield  {journal} {\bibinfo  {journal}
  {Nature}\ }\textbf {\bibinfo {volume} {526}},\ \bibinfo {pages} {682}
  (\bibinfo {year} {2015})}\BibitemShut {NoStop}%
\bibitem [{\citenamefont {Shalm}\ \emph {et~al.}(2015)\citenamefont {Shalm},
  \citenamefont {Meyer-Scott}, \citenamefont {Christensen}, \citenamefont
  {Bierhorst}, \citenamefont {Wayne}, \citenamefont {Stevens}, \citenamefont
  {Gerrits}, \citenamefont {Glancy}, \citenamefont {Hamel}, \citenamefont
  {Allman}, \citenamefont {Coakley}, \citenamefont {Dyer}, \citenamefont
  {Hodge}, \citenamefont {Lita}, \citenamefont {Verma}, \citenamefont
  {Lambrocco}, \citenamefont {Tortorici}, \citenamefont {Migdall},
  \citenamefont {Zhang}, \citenamefont {Kumor}, \citenamefont {Farr},
  \citenamefont {Marsili}, \citenamefont {Shaw}, \citenamefont {Stern},
  \citenamefont {Abell\'an}, \citenamefont {Amaya}, \citenamefont {Pruneri},
  \citenamefont {Jennewein}, \citenamefont {Mitchell}, \citenamefont {Kwiat},
  \citenamefont {Bienfang}, \citenamefont {Mirin}, \citenamefont {Knill},\ and\
  \citenamefont {Nam}}]{Shalm15}%
  \BibitemOpen
  \bibfield  {author} {\bibinfo {author} {\bibfnamefont {L.~K.}\ \bibnamefont
  {Shalm}}, \bibinfo {author} {\bibfnamefont {E.}~\bibnamefont {Meyer-Scott}},
  \bibinfo {author} {\bibfnamefont {B.~G.}\ \bibnamefont {Christensen}},
  \bibinfo {author} {\bibfnamefont {P.}~\bibnamefont {Bierhorst}}, \bibinfo
  {author} {\bibfnamefont {M.~A.}\ \bibnamefont {Wayne}}, \bibinfo {author}
  {\bibfnamefont {M.~J.}\ \bibnamefont {Stevens}}, \bibinfo {author}
  {\bibfnamefont {T.}~\bibnamefont {Gerrits}}, \bibinfo {author} {\bibfnamefont
  {S.}~\bibnamefont {Glancy}}, \bibinfo {author} {\bibfnamefont {D.~R.}\
  \bibnamefont {Hamel}}, \bibinfo {author} {\bibfnamefont {M.~S.}\ \bibnamefont
  {Allman}}, \bibinfo {author} {\bibfnamefont {K.~J.}\ \bibnamefont {Coakley}},
  \bibinfo {author} {\bibfnamefont {S.~D.}\ \bibnamefont {Dyer}}, \bibinfo
  {author} {\bibfnamefont {C.}~\bibnamefont {Hodge}}, \bibinfo {author}
  {\bibfnamefont {A.~E.}\ \bibnamefont {Lita}}, \bibinfo {author}
  {\bibfnamefont {V.~B.}\ \bibnamefont {Verma}}, \bibinfo {author}
  {\bibfnamefont {C.}~\bibnamefont {Lambrocco}}, \bibinfo {author}
  {\bibfnamefont {E.}~\bibnamefont {Tortorici}}, \bibinfo {author}
  {\bibfnamefont {A.~L.}\ \bibnamefont {Migdall}}, \bibinfo {author}
  {\bibfnamefont {Y.}~\bibnamefont {Zhang}}, \bibinfo {author} {\bibfnamefont
  {D.~R.}\ \bibnamefont {Kumor}}, \bibinfo {author} {\bibfnamefont {W.~H.}\
  \bibnamefont {Farr}}, \bibinfo {author} {\bibfnamefont {F.}~\bibnamefont
  {Marsili}}, \bibinfo {author} {\bibfnamefont {M.~D.}\ \bibnamefont {Shaw}},
  \bibinfo {author} {\bibfnamefont {J.~A.}\ \bibnamefont {Stern}}, \bibinfo
  {author} {\bibfnamefont {C.}~\bibnamefont {Abell\'an}}, \bibinfo {author}
  {\bibfnamefont {W.}~\bibnamefont {Amaya}}, \bibinfo {author} {\bibfnamefont
  {V.}~\bibnamefont {Pruneri}}, \bibinfo {author} {\bibfnamefont
  {T.}~\bibnamefont {Jennewein}}, \bibinfo {author} {\bibfnamefont {M.~W.}\
  \bibnamefont {Mitchell}}, \bibinfo {author} {\bibfnamefont {P.~G.}\
  \bibnamefont {Kwiat}}, \bibinfo {author} {\bibfnamefont {J.~C.}\ \bibnamefont
  {Bienfang}}, \bibinfo {author} {\bibfnamefont {R.~P.}\ \bibnamefont {Mirin}},
  \bibinfo {author} {\bibfnamefont {E.}~\bibnamefont {Knill}}, \ and\ \bibinfo
  {author} {\bibfnamefont {S.~W.}\ \bibnamefont {Nam}},\ }\href {\doibase
  10.1103/PhysRevLett.115.250402} {\bibfield  {journal} {\bibinfo  {journal}
  {Phys. Rev. Lett.}\ }\textbf {\bibinfo {volume} {115}},\ \bibinfo {pages}
  {250402} (\bibinfo {year} {2015})}\BibitemShut {NoStop}%
\bibitem [{\citenamefont {Giustina}\ \emph {et~al.}(2015)\citenamefont
  {Giustina}, \citenamefont {Versteegh}, \citenamefont {Wengerowsky},
  \citenamefont {Handsteiner}, \citenamefont {Hochrainer}, \citenamefont
  {Phelan}, \citenamefont {Steinlechner}, \citenamefont {Kofler}, \citenamefont
  {Larsson}, \citenamefont {Abell\'an}, \citenamefont {Amaya}, \citenamefont
  {Pruneri}, \citenamefont {Mitchell}, \citenamefont {Beyer}, \citenamefont
  {Gerrits}, \citenamefont {Lita}, \citenamefont {Shalm}, \citenamefont {Nam},
  \citenamefont {Scheidl}, \citenamefont {Ursin}, \citenamefont {Wittmann},\
  and\ \citenamefont {Zeilinger}}]{Giustina15}%
  \BibitemOpen
  \bibfield  {author} {\bibinfo {author} {\bibfnamefont {M.}~\bibnamefont
  {Giustina}}, \bibinfo {author} {\bibfnamefont {M.~A.~M.}\ \bibnamefont
  {Versteegh}}, \bibinfo {author} {\bibfnamefont {S.}~\bibnamefont
  {Wengerowsky}}, \bibinfo {author} {\bibfnamefont {J.}~\bibnamefont
  {Handsteiner}}, \bibinfo {author} {\bibfnamefont {A.}~\bibnamefont
  {Hochrainer}}, \bibinfo {author} {\bibfnamefont {K.}~\bibnamefont {Phelan}},
  \bibinfo {author} {\bibfnamefont {F.}~\bibnamefont {Steinlechner}}, \bibinfo
  {author} {\bibfnamefont {J.}~\bibnamefont {Kofler}}, \bibinfo {author}
  {\bibfnamefont {J.-A.}\ \bibnamefont {Larsson}}, \bibinfo {author}
  {\bibfnamefont {C.}~\bibnamefont {Abell\'an}}, \bibinfo {author}
  {\bibfnamefont {W.}~\bibnamefont {Amaya}}, \bibinfo {author} {\bibfnamefont
  {V.}~\bibnamefont {Pruneri}}, \bibinfo {author} {\bibfnamefont {M.~W.}\
  \bibnamefont {Mitchell}}, \bibinfo {author} {\bibfnamefont {J.}~\bibnamefont
  {Beyer}}, \bibinfo {author} {\bibfnamefont {T.}~\bibnamefont {Gerrits}},
  \bibinfo {author} {\bibfnamefont {A.~E.}\ \bibnamefont {Lita}}, \bibinfo
  {author} {\bibfnamefont {L.~K.}\ \bibnamefont {Shalm}}, \bibinfo {author}
  {\bibfnamefont {S.~W.}\ \bibnamefont {Nam}}, \bibinfo {author} {\bibfnamefont
  {T.}~\bibnamefont {Scheidl}}, \bibinfo {author} {\bibfnamefont
  {R.}~\bibnamefont {Ursin}}, \bibinfo {author} {\bibfnamefont
  {B.}~\bibnamefont {Wittmann}}, \ and\ \bibinfo {author} {\bibfnamefont
  {A.}~\bibnamefont {Zeilinger}},\ }\href {\doibase
  10.1103/PhysRevLett.115.250401} {\bibfield  {journal} {\bibinfo  {journal}
  {Phys. Rev. Lett.}\ }\textbf {\bibinfo {volume} {115}},\ \bibinfo {pages}
  {250401} (\bibinfo {year} {2015})}\BibitemShut {NoStop}%
\bibitem [{\citenamefont {Branciard}\ \emph {et~al.}(2013)\citenamefont
  {Branciard}, \citenamefont {Rosset}, \citenamefont {Liang},\ and\
  \citenamefont {Gisin}}]{Branciard13}%
  \BibitemOpen
  \bibfield  {author} {\bibinfo {author} {\bibfnamefont {C.}~\bibnamefont
  {Branciard}}, \bibinfo {author} {\bibfnamefont {D.}~\bibnamefont {Rosset}},
  \bibinfo {author} {\bibfnamefont {Y.-C.}\ \bibnamefont {Liang}}, \ and\
  \bibinfo {author} {\bibfnamefont {N.}~\bibnamefont {Gisin}},\ }\href
  {\doibase 10.1103/PhysRevLett.110.060405} {\bibfield  {journal} {\bibinfo
  {journal} {Phys. Rev. Lett.}\ }\textbf {\bibinfo {volume} {110}},\ \bibinfo
  {pages} {060405} (\bibinfo {year} {2013})}\BibitemShut {NoStop}%
\bibitem [{\citenamefont {Nawareg}\ \emph {et~al.}(2015)\citenamefont
  {Nawareg}, \citenamefont {Muhammad}, \citenamefont {Amselem},\ and\
  \citenamefont {Bourennane}}]{nawareg2015experimental}%
  \BibitemOpen
  \bibfield  {author} {\bibinfo {author} {\bibfnamefont {M.}~\bibnamefont
  {Nawareg}}, \bibinfo {author} {\bibfnamefont {S.}~\bibnamefont {Muhammad}},
  \bibinfo {author} {\bibfnamefont {E.}~\bibnamefont {Amselem}}, \ and\
  \bibinfo {author} {\bibfnamefont {M.}~\bibnamefont {Bourennane}},\
  }\href@noop {} {\bibfield  {journal} {\bibinfo  {journal} {Scientific
  reports}\ }\textbf {\bibinfo {volume} {5}} (\bibinfo {year}
  {2015})}\BibitemShut {NoStop}%
\bibitem [{\citenamefont {Lo}\ \emph {et~al.}(2012)\citenamefont {Lo},
  \citenamefont {Curty},\ and\ \citenamefont {Qi}}]{Lo12}%
  \BibitemOpen
  \bibfield  {author} {\bibinfo {author} {\bibfnamefont {H.-K.}\ \bibnamefont
  {Lo}}, \bibinfo {author} {\bibfnamefont {M.}~\bibnamefont {Curty}}, \ and\
  \bibinfo {author} {\bibfnamefont {B.}~\bibnamefont {Qi}},\ }\href {\doibase
  10.1103/PhysRevLett.108.130503} {\bibfield  {journal} {\bibinfo  {journal}
  {Phys. Rev. Lett.}\ }\textbf {\bibinfo {volume} {108}},\ \bibinfo {pages}
  {130503} (\bibinfo {year} {2012})}\BibitemShut {NoStop}%
\bibitem [{\citenamefont {Buscemi}(2012)}]{Buscemi12}%
  \BibitemOpen
  \bibfield  {author} {\bibinfo {author} {\bibfnamefont {F.}~\bibnamefont
  {Buscemi}},\ }\href {\doibase 10.1103/PhysRevLett.108.200401} {\bibfield
  {journal} {\bibinfo  {journal} {Phys. Rev. Lett.}\ }\textbf {\bibinfo
  {volume} {108}},\ \bibinfo {pages} {200401} (\bibinfo {year}
  {2012})}\BibitemShut {NoStop}%
\bibitem [{\citenamefont {G{\"u}hne}\ \emph {et~al.}(2005)\citenamefont
  {G{\"u}hne}, \citenamefont {T{\'o}th},\ and\ \citenamefont
  {Briegel}}]{guhne2005multipartite}%
  \BibitemOpen
  \bibfield  {author} {\bibinfo {author} {\bibfnamefont {O.}~\bibnamefont
  {G{\"u}hne}}, \bibinfo {author} {\bibfnamefont {G.}~\bibnamefont {T{\'o}th}},
  \ and\ \bibinfo {author} {\bibfnamefont {H.~J.}\ \bibnamefont {Briegel}},\
  }\href@noop {} {\bibfield  {journal} {\bibinfo  {journal} {New Journal of
  Physics}\ }\textbf {\bibinfo {volume} {7}},\ \bibinfo {pages} {229} (\bibinfo
  {year} {2005})}\BibitemShut {NoStop}%
\bibitem [{\citenamefont {Seevinck}\ and\ \citenamefont
  {Uffink}(2001)}]{seevinck2001sufficient}%
  \BibitemOpen
  \bibfield  {author} {\bibinfo {author} {\bibfnamefont {M.}~\bibnamefont
  {Seevinck}}\ and\ \bibinfo {author} {\bibfnamefont {J.}~\bibnamefont
  {Uffink}},\ }\href@noop {} {\bibfield  {journal} {\bibinfo  {journal}
  {Physical Review A}\ }\textbf {\bibinfo {volume} {65}},\ \bibinfo {pages}
  {012107} (\bibinfo {year} {2001})}\BibitemShut {NoStop}%
\bibitem [{\citenamefont {Giovannetti}\ \emph {et~al.}(2004)\citenamefont
  {Giovannetti}, \citenamefont {Lloyd},\ and\ \citenamefont
  {Maccone}}]{giovannetti2004quantum}%
  \BibitemOpen
  \bibfield  {author} {\bibinfo {author} {\bibfnamefont {V.}~\bibnamefont
  {Giovannetti}}, \bibinfo {author} {\bibfnamefont {S.}~\bibnamefont {Lloyd}},
  \ and\ \bibinfo {author} {\bibfnamefont {L.}~\bibnamefont {Maccone}},\
  }\href@noop {} {\bibfield  {journal} {\bibinfo  {journal} {Science}\ }\textbf
  {\bibinfo {volume} {306}},\ \bibinfo {pages} {1330} (\bibinfo {year}
  {2004})}\BibitemShut {NoStop}%
\bibitem [{\citenamefont {S{\o}rensen}\ and\ \citenamefont
  {M{\o}lmer}(2001)}]{sorensen2001entanglement}%
  \BibitemOpen
  \bibfield  {author} {\bibinfo {author} {\bibfnamefont {A.~S.}\ \bibnamefont
  {S{\o}rensen}}\ and\ \bibinfo {author} {\bibfnamefont {K.}~\bibnamefont
  {M{\o}lmer}},\ }\href@noop {} {\bibfield  {journal} {\bibinfo  {journal}
  {Physical review letters}\ }\textbf {\bibinfo {volume} {86}},\ \bibinfo
  {pages} {4431} (\bibinfo {year} {2001})}\BibitemShut {NoStop}%
\bibitem [{\citenamefont {Huber}\ and\ \citenamefont
  {de~Vicente}(2013)}]{huber2013structure}%
  \BibitemOpen
  \bibfield  {author} {\bibinfo {author} {\bibfnamefont {M.}~\bibnamefont
  {Huber}}\ and\ \bibinfo {author} {\bibfnamefont {J.~I.}\ \bibnamefont
  {de~Vicente}},\ }\href@noop {} {\bibfield  {journal} {\bibinfo  {journal}
  {Physical review letters}\ }\textbf {\bibinfo {volume} {110}},\ \bibinfo
  {pages} {030501} (\bibinfo {year} {2013})}\BibitemShut {NoStop}%
\bibitem [{\citenamefont {Shahandeh}\ \emph {et~al.}(2014)\citenamefont
  {Shahandeh}, \citenamefont {Sperling},\ and\ \citenamefont
  {Vogel}}]{shahandeh2014structural}%
  \BibitemOpen
  \bibfield  {author} {\bibinfo {author} {\bibfnamefont {F.}~\bibnamefont
  {Shahandeh}}, \bibinfo {author} {\bibfnamefont {J.}~\bibnamefont {Sperling}},
  \ and\ \bibinfo {author} {\bibfnamefont {W.}~\bibnamefont {Vogel}},\
  }\href@noop {} {\bibfield  {journal} {\bibinfo  {journal} {Physical review
  letters}\ }\textbf {\bibinfo {volume} {113}},\ \bibinfo {pages} {260502}
  (\bibinfo {year} {2014})}\BibitemShut {NoStop}%
\bibitem [{\citenamefont {Liang}\ \emph {et~al.}(2015)\citenamefont {Liang},
  \citenamefont {Rosset}, \citenamefont {Bancal}, \citenamefont {P{\"u}tz},
  \citenamefont {Barnea},\ and\ \citenamefont {Gisin}}]{liang2015family}%
  \BibitemOpen
  \bibfield  {author} {\bibinfo {author} {\bibfnamefont {Y.-C.}\ \bibnamefont
  {Liang}}, \bibinfo {author} {\bibfnamefont {D.}~\bibnamefont {Rosset}},
  \bibinfo {author} {\bibfnamefont {J.-D.}\ \bibnamefont {Bancal}}, \bibinfo
  {author} {\bibfnamefont {G.}~\bibnamefont {P{\"u}tz}}, \bibinfo {author}
  {\bibfnamefont {T.~J.}\ \bibnamefont {Barnea}}, \ and\ \bibinfo {author}
  {\bibfnamefont {N.}~\bibnamefont {Gisin}},\ }\href@noop {} {\bibfield
  {journal} {\bibinfo  {journal} {Physical review letters}\ }\textbf {\bibinfo
  {volume} {114}},\ \bibinfo {pages} {190401} (\bibinfo {year}
  {2015})}\BibitemShut {NoStop}%
\bibitem [{\citenamefont {Chru{\'s}ci{\'n}ski}\ and\ \citenamefont
  {Sarbicki}(2014)}]{Dariusz14}%
  \BibitemOpen
  \bibfield  {author} {\bibinfo {author} {\bibfnamefont {D.}~\bibnamefont
  {Chru{\'s}ci{\'n}ski}}\ and\ \bibinfo {author} {\bibfnamefont
  {G.}~\bibnamefont {Sarbicki}},\ }\href
  {http://stacks.iop.org/1751-8121/47/i=48/a=483001} {\bibfield  {journal}
  {\bibinfo  {journal} {Journal of Physics A: Mathematical and Theoretical}\
  }\textbf {\bibinfo {volume} {47}},\ \bibinfo {pages} {483001} (\bibinfo
  {year} {2014})}\BibitemShut {NoStop}%
\bibitem [{\citenamefont {Horodecki}\ \emph {et~al.}(1996)\citenamefont
  {Horodecki}, \citenamefont {Horodecki},\ and\ \citenamefont
  {Horodecki}}]{horodecki1996separability}%
  \BibitemOpen
  \bibfield  {author} {\bibinfo {author} {\bibfnamefont {M.}~\bibnamefont
  {Horodecki}}, \bibinfo {author} {\bibfnamefont {P.}~\bibnamefont
  {Horodecki}}, \ and\ \bibinfo {author} {\bibfnamefont {R.}~\bibnamefont
  {Horodecki}},\ }\href@noop {} {\bibfield  {journal} {\bibinfo  {journal}
  {Physics Letters A}\ }\textbf {\bibinfo {volume} {223}},\ \bibinfo {pages}
  {1} (\bibinfo {year} {1996})}\BibitemShut {NoStop}%
\bibitem [{\citenamefont {Werner}(1989)}]{PhysRevA.40.4277}%
  \BibitemOpen
  \bibfield  {author} {\bibinfo {author} {\bibfnamefont {R.~F.}\ \bibnamefont
  {Werner}},\ }\href {\doibase 10.1103/PhysRevA.40.4277} {\bibfield  {journal}
  {\bibinfo  {journal} {Phys. Rev. A}\ }\textbf {\bibinfo {volume} {40}},\
  \bibinfo {pages} {4277} (\bibinfo {year} {1989})}\BibitemShut {NoStop}%
\bibitem [{\citenamefont {Sperling}\ and\ \citenamefont
  {Vogel}(2013)}]{sperling2013multipartite}%
  \BibitemOpen
  \bibfield  {author} {\bibinfo {author} {\bibfnamefont {J.}~\bibnamefont
  {Sperling}}\ and\ \bibinfo {author} {\bibfnamefont {W.}~\bibnamefont
  {Vogel}},\ }\href@noop {} {\bibfield  {journal} {\bibinfo  {journal}
  {Physical review letters}\ }\textbf {\bibinfo {volume} {111}},\ \bibinfo
  {pages} {110503} (\bibinfo {year} {2013})}\BibitemShut {NoStop}%
\bibitem [{\citenamefont {Brunner}\ \emph {et~al.}(2012)\citenamefont
  {Brunner}, \citenamefont {Sharam},\ and\ \citenamefont
  {Vertesi}}]{brunner2012testing}%
  \BibitemOpen
  \bibfield  {author} {\bibinfo {author} {\bibfnamefont {N.}~\bibnamefont
  {Brunner}}, \bibinfo {author} {\bibfnamefont {J.}~\bibnamefont {Sharam}}, \
  and\ \bibinfo {author} {\bibfnamefont {T.}~\bibnamefont {Vertesi}},\
  }\href@noop {} {\bibfield  {journal} {\bibinfo  {journal} {Physical review
  letters}\ }\textbf {\bibinfo {volume} {108}},\ \bibinfo {pages} {110501}
  (\bibinfo {year} {2012})}\BibitemShut {NoStop}%
\bibitem [{\citenamefont {Yuan}\ \emph {et~al.}(2016)\citenamefont {Yuan},
  \citenamefont {Mei}, \citenamefont {Zhou},\ and\ \citenamefont
  {Ma}}]{PhysRevA.93.042317}%
  \BibitemOpen
  \bibfield  {author} {\bibinfo {author} {\bibfnamefont {X.}~\bibnamefont
  {Yuan}}, \bibinfo {author} {\bibfnamefont {Q.}~\bibnamefont {Mei}}, \bibinfo
  {author} {\bibfnamefont {S.}~\bibnamefont {Zhou}}, \ and\ \bibinfo {author}
  {\bibfnamefont {X.}~\bibnamefont {Ma}},\ }\href {\doibase
  10.1103/PhysRevA.93.042317} {\bibfield  {journal} {\bibinfo  {journal} {Phys.
  Rev. A}\ }\textbf {\bibinfo {volume} {93}},\ \bibinfo {pages} {042317}
  (\bibinfo {year} {2016})}\BibitemShut {NoStop}%
\bibitem [{\citenamefont {Verbanis}\ \emph {et~al.}(2016)\citenamefont
  {Verbanis}, \citenamefont {Martin}, \citenamefont {Rosset}, \citenamefont
  {Lim}, \citenamefont {Thew},\ and\ \citenamefont
  {Zbinden}}]{PhysRevLett.116.190501}%
  \BibitemOpen
  \bibfield  {author} {\bibinfo {author} {\bibfnamefont {E.}~\bibnamefont
  {Verbanis}}, \bibinfo {author} {\bibfnamefont {A.}~\bibnamefont {Martin}},
  \bibinfo {author} {\bibfnamefont {D.}~\bibnamefont {Rosset}}, \bibinfo
  {author} {\bibfnamefont {C.~C.~W.}\ \bibnamefont {Lim}}, \bibinfo {author}
  {\bibfnamefont {R.~T.}\ \bibnamefont {Thew}}, \ and\ \bibinfo {author}
  {\bibfnamefont {H.}~\bibnamefont {Zbinden}},\ }\href {\doibase
  10.1103/PhysRevLett.116.190501} {\bibfield  {journal} {\bibinfo  {journal}
  {Phys. Rev. Lett.}\ }\textbf {\bibinfo {volume} {116}},\ \bibinfo {pages}
  {190501} (\bibinfo {year} {2016})}\BibitemShut {NoStop}%
\bibitem [{\citenamefont {Rosset}\ \emph {et~al.}(2013)\citenamefont {Rosset},
  \citenamefont {Branciard}, \citenamefont {Gisin},\ and\ \citenamefont
  {Liang}}]{rosset2013entangled}%
  \BibitemOpen
  \bibfield  {author} {\bibinfo {author} {\bibfnamefont {D.}~\bibnamefont
  {Rosset}}, \bibinfo {author} {\bibfnamefont {C.}~\bibnamefont {Branciard}},
  \bibinfo {author} {\bibfnamefont {N.}~\bibnamefont {Gisin}}, \ and\ \bibinfo
  {author} {\bibfnamefont {Y.-C.}\ \bibnamefont {Liang}},\ }\href@noop {}
  {\bibfield  {journal} {\bibinfo  {journal} {New Journal of Physics}\ }\textbf
  {\bibinfo {volume} {15}},\ \bibinfo {pages} {053025} (\bibinfo {year}
  {2013})}\BibitemShut {NoStop}%
\bibitem [{\citenamefont {Brunner}\ \emph {et~al.}(2014)\citenamefont
  {Brunner}, \citenamefont {Cavalcanti}, \citenamefont {Pironio}, \citenamefont
  {Scarani},\ and\ \citenamefont {Wehner}}]{Brunner14}%
  \BibitemOpen
  \bibfield  {author} {\bibinfo {author} {\bibfnamefont {N.}~\bibnamefont
  {Brunner}}, \bibinfo {author} {\bibfnamefont {D.}~\bibnamefont {Cavalcanti}},
  \bibinfo {author} {\bibfnamefont {S.}~\bibnamefont {Pironio}}, \bibinfo
  {author} {\bibfnamefont {V.}~\bibnamefont {Scarani}}, \ and\ \bibinfo
  {author} {\bibfnamefont {S.}~\bibnamefont {Wehner}},\ }\href {\doibase
  10.1103/RevModPhys.86.419} {\bibfield  {journal} {\bibinfo  {journal} {Rev.
  Mod. Phys.}\ }\textbf {\bibinfo {volume} {86}},\ \bibinfo {pages} {419}
  (\bibinfo {year} {2014})}\BibitemShut {NoStop}%
\bibitem [{\citenamefont {Massar}\ and\ \citenamefont
  {Pironio}(2003)}]{Massar03}%
  \BibitemOpen
  \bibfield  {author} {\bibinfo {author} {\bibfnamefont {S.}~\bibnamefont
  {Massar}}\ and\ \bibinfo {author} {\bibfnamefont {S.}~\bibnamefont
  {Pironio}},\ }\href {\doibase 10.1103/PhysRevA.68.062109} {\bibfield
  {journal} {\bibinfo  {journal} {Phys. Rev. A}\ }\textbf {\bibinfo {volume}
  {68}},\ \bibinfo {pages} {062109} (\bibinfo {year} {2003})}\BibitemShut
  {NoStop}%
\bibitem [{\citenamefont {Wilms}\ \emph {et~al.}(2008)\citenamefont {Wilms},
  \citenamefont {Disser}, \citenamefont {Alber},\ and\ \citenamefont
  {Percival}}]{Wilms08}%
  \BibitemOpen
  \bibfield  {author} {\bibinfo {author} {\bibfnamefont {J.}~\bibnamefont
  {Wilms}}, \bibinfo {author} {\bibfnamefont {Y.}~\bibnamefont {Disser}},
  \bibinfo {author} {\bibfnamefont {G.}~\bibnamefont {Alber}}, \ and\ \bibinfo
  {author} {\bibfnamefont {I.~C.}\ \bibnamefont {Percival}},\ }\href {\doibase
  10.1103/PhysRevA.78.032116} {\bibfield  {journal} {\bibinfo  {journal} {Phys.
  Rev. A}\ }\textbf {\bibinfo {volume} {78}},\ \bibinfo {pages} {032116}
  (\bibinfo {year} {2008})}\BibitemShut {NoStop}%
\bibitem [{\citenamefont {Koh}\ \emph {et~al.}(2012)\citenamefont {Koh},
  \citenamefont {Hall}, \citenamefont {Setiawan}, \citenamefont {Pope},
  \citenamefont {Marletto}, \citenamefont {Kay}, \citenamefont {Scarani},\ and\
  \citenamefont {Ekert}}]{Koh12}%
  \BibitemOpen
  \bibfield  {author} {\bibinfo {author} {\bibfnamefont {D.~E.}\ \bibnamefont
  {Koh}}, \bibinfo {author} {\bibfnamefont {M.~J.~W.}\ \bibnamefont {Hall}},
  \bibinfo {author} {\bibnamefont {Setiawan}}, \bibinfo {author} {\bibfnamefont
  {J.~E.}\ \bibnamefont {Pope}}, \bibinfo {author} {\bibfnamefont
  {C.}~\bibnamefont {Marletto}}, \bibinfo {author} {\bibfnamefont
  {A.}~\bibnamefont {Kay}}, \bibinfo {author} {\bibfnamefont {V.}~\bibnamefont
  {Scarani}}, \ and\ \bibinfo {author} {\bibfnamefont {A.}~\bibnamefont
  {Ekert}},\ }\href {\doibase 10.1103/PhysRevLett.109.160404} {\bibfield
  {journal} {\bibinfo  {journal} {Phys. Rev. Lett.}\ }\textbf {\bibinfo
  {volume} {109}},\ \bibinfo {pages} {160404} (\bibinfo {year}
  {2012})}\BibitemShut {NoStop}%
\bibitem [{\citenamefont {Yuan}\ \emph
  {et~al.}(2015{\natexlab{a}})\citenamefont {Yuan}, \citenamefont {Cao},\ and\
  \citenamefont {Ma}}]{Yuan15}%
  \BibitemOpen
  \bibfield  {author} {\bibinfo {author} {\bibfnamefont {X.}~\bibnamefont
  {Yuan}}, \bibinfo {author} {\bibfnamefont {Z.}~\bibnamefont {Cao}}, \ and\
  \bibinfo {author} {\bibfnamefont {X.}~\bibnamefont {Ma}},\ }\href {\doibase
  10.1103/PhysRevA.91.032111} {\bibfield  {journal} {\bibinfo  {journal} {Phys.
  Rev. A}\ }\textbf {\bibinfo {volume} {91}},\ \bibinfo {pages} {032111}
  (\bibinfo {year} {2015}{\natexlab{a}})}\BibitemShut {NoStop}%
\bibitem [{\citenamefont {Yuan}\ \emph
  {et~al.}(2015{\natexlab{b}})\citenamefont {Yuan}, \citenamefont {Zhao},\ and\
  \citenamefont {Ma}}]{Yuan152}%
  \BibitemOpen
  \bibfield  {author} {\bibinfo {author} {\bibfnamefont {X.}~\bibnamefont
  {Yuan}}, \bibinfo {author} {\bibfnamefont {Q.}~\bibnamefont {Zhao}}, \ and\
  \bibinfo {author} {\bibfnamefont {X.}~\bibnamefont {Ma}},\ }\href {\doibase
  10.1103/PhysRevA.92.022107} {\bibfield  {journal} {\bibinfo  {journal} {Phys.
  Rev. A}\ }\textbf {\bibinfo {volume} {92}},\ \bibinfo {pages} {022107}
  (\bibinfo {year} {2015}{\natexlab{b}})}\BibitemShut {NoStop}%
\bibitem [{\citenamefont {D\"ur}\ \emph {et~al.}(2000)\citenamefont {D\"ur},
  \citenamefont {Vidal},\ and\ \citenamefont {Cirac}}]{Dur00}%
  \BibitemOpen
  \bibfield  {author} {\bibinfo {author} {\bibfnamefont {W.}~\bibnamefont
  {D\"ur}}, \bibinfo {author} {\bibfnamefont {G.}~\bibnamefont {Vidal}}, \ and\
  \bibinfo {author} {\bibfnamefont {J.~I.}\ \bibnamefont {Cirac}},\ }\href
  {\doibase 10.1103/PhysRevA.62.062314} {\bibfield  {journal} {\bibinfo
  {journal} {Phys. Rev. A}\ }\textbf {\bibinfo {volume} {62}},\ \bibinfo
  {pages} {062314} (\bibinfo {year} {2000})}\BibitemShut {NoStop}%
\bibitem [{\citenamefont {Borsten}\ \emph {et~al.}(2010)\citenamefont
  {Borsten}, \citenamefont {Dahanayake}, \citenamefont {Duff}, \citenamefont
  {Marrani},\ and\ \citenamefont {Rubens}}]{borsten2010four}%
  \BibitemOpen
  \bibfield  {author} {\bibinfo {author} {\bibfnamefont {L.}~\bibnamefont
  {Borsten}}, \bibinfo {author} {\bibfnamefont {D.}~\bibnamefont {Dahanayake}},
  \bibinfo {author} {\bibfnamefont {M.~J.}\ \bibnamefont {Duff}}, \bibinfo
  {author} {\bibfnamefont {A.}~\bibnamefont {Marrani}}, \ and\ \bibinfo
  {author} {\bibfnamefont {W.}~\bibnamefont {Rubens}},\ }\href@noop {}
  {\bibfield  {journal} {\bibinfo  {journal} {Physical review letters}\
  }\textbf {\bibinfo {volume} {105}},\ \bibinfo {pages} {100507} (\bibinfo
  {year} {2010})}\BibitemShut {NoStop}%
\bibitem [{\citenamefont {Acin}\ \emph {et~al.}(2001)\citenamefont {Acin},
  \citenamefont {Bru{\ss}}, \citenamefont {Lewenstein},\ and\ \citenamefont
  {Sanpera}}]{acin2001classification}%
  \BibitemOpen
  \bibfield  {author} {\bibinfo {author} {\bibfnamefont {A.}~\bibnamefont
  {Acin}}, \bibinfo {author} {\bibfnamefont {D.}~\bibnamefont {Bru{\ss}}},
  \bibinfo {author} {\bibfnamefont {M.}~\bibnamefont {Lewenstein}}, \ and\
  \bibinfo {author} {\bibfnamefont {A.}~\bibnamefont {Sanpera}},\ }\href@noop
  {} {\bibfield  {journal} {\bibinfo  {journal} {Physical Review Letters}\
  }\textbf {\bibinfo {volume} {87}},\ \bibinfo {pages} {040401} (\bibinfo
  {year} {2001})}\BibitemShut {NoStop}%
\end{thebibliography}%


\end{document}